\documentclass[a4paper,11pt]{article}
\usepackage{natbib,amsmath,amssymb,amsthm,latexsym,graphicx}

\newtheorem{theorem}{Theorem}[section]

\theoremstyle{remark}

\title{A Wilcoxon--Mann--Whitney type test \\ for infinite dimensional data}
\author{\vspace{0.3in} Anirvan Chakraborty ~and Probal Chaudhuri}
\date{}
\begin{document}

\maketitle
\vspace{-0.6in}
\begin{center}
 Theoretical Statistics and Mathematics Unit, \\
 Indian Statistical Institute \\ 
 203, B. T. Road, Kolkata - 700108, INDIA. \\
 emails: anirvan\_r@isical.ac.in, probal@isical.ac.in
\end{center}
\vspace{0.15in}
\begin{abstract}
The Wilcoxon--Mann--Whitney test is a robust competitor of the t-test in the univariate setting. For finite dimensional multivariate data, several extensions of the Wilcoxon--Mann--Whitney test have been shown to have better performance than Hotelling's $T^{2}$ test for many non-Gaussian distributions of the data. In this paper, we study a Wilcoxon--Mann--Whitney type test based on spatial ranks for data in infinite dimensional spaces. We demonstrate the performance of this test using some real and simulated datasets. We also investigate the asymptotic properties of the proposed test and compare the test with a wide range of competing tests.
\vspace{0.1in} \\
\textbf{Keywords}: {Brownian motion, functional data, G\^ateaux derivative, smooth Banach spaces, spatial rank, t processes, two-sample problem, U-statistics.}
\end{abstract}

\section{Introduction} 
\label{1} 
\indent For univariate data, the Wilcoxon--Mann--Whitney test is known to have better power than the t-test for several non-Gaussian distributions (see, e.g., \cite{HSS99}). Various extensions of the Wilcoxon--Mann--Whitney test have been studied for multivariate data in finite dimensional spaces (see, e.g., \cite{PS71}, \cite{RP90}, \cite{LS93}, \cite{CM97}, \cite{CC99} and \cite{Oja99}), and these extensions too outperform Hotelling's $T^{2}$ test for a number of non-Gaussian multivariate distributions. Nowadays, we often have to analyze data, which are curves or functions, e.g., the ECG curves of patients, the temperature curves of different regions, the spectrometry readings over a range of wavelengths etc. Such data, popularly known as functional data, can be conveniently handled by viewing them as random observations from probability distributions in infinite dimensional spaces, e.g., the space of functions on an interval. For testing the equality of means of two functional datasets, \cite{HKR13} proposed two test statistics based on orthogonal projections of the difference between the sample mean functions. One of those statistics is same as Hotelling's $T^{2}$ statistic based on a finite number of such projections. \cite{CFF04} and \cite{ZC07} studied two $L_{2}$-norm based tests for functional analysis of variance and functional linear models, respectively. For the problem of testing the equality of two mean functions, these two statistics reduce to a constant multiple of the $L_{2}$-norm of the difference between the sample mean functions. A two sample test for the equality of the means based on this latter statistic was studied by \cite{ZPZ10}. In a different direction, \cite{BS96}, \cite{FL98}, \cite{CQ10} and \cite{SKK13} studied some tests for comparing the means of two finite dimensional datasets for which the data dimension is larger than the sample size, and it grows with the sample size. These authors worked in a setup, which is different from the infinite dimensional setup considered in this paper. Consequently, the tests and the results obtained by these authors are quite different from ours. All of the above-mentioned tests for functional and high dimensional data perform poorly when the observations have non-Gaussian distributions with heavy-tails. \\
\indent Some of the Wilcoxon--Mann--Whitney type tests for finite dimensional data in $\mathbb{R}^{d}$, e.g., those defined using simplices (see, e.g., \cite{LS93} and \cite{Oja99}) or those based on interdirections (see, e.g., \cite{RP90}), cannot be extended into infinite dimensional spaces due to their dependence on the finite dimensional coordinate system in $\mathbb{R}^{d}$. Further, a test that involves standardization by some covariance matrix computed from the sample (see, e.g., \cite{PS71} and \cite{Oja99}) cannot be used due to the singularity of such a sample covariance matrix when the data dimension exceeds the sample size. \\
\indent Many of the function spaces, where functional data lie, are infinite dimensional Banach spaces. In Section $2$, we develop a Wilcoxon--Mann--Whitney type test based on spatial ranks for data lying in those spaces. We show that the proposed test statistic has an asymptotic Gaussian distribution. We implement the test using this asymptotic distribution and demonstrate its performance using some real benchmark data. In Section $3$, we derive the asymptotic distribution of our Wilcoxon--Mann--Whitney type test statistic under some sequences of shrinking location shift models. We carry out an asymptotic power comparison between our test and several other tests for functional data. It is observed that our Wilcoxon--Mann--Whitney type test has superior performance than these competing tests in most of the heavy-tailed models as well as in some of the Gaussian models considered. In Section $4$, we report the results from a detailed simulation study comparing the finite sample performance of our test with that of a wide range of two sample tests available in the literature for infinite dimensional data.

\section{The construction and the implementation of the test}
\label{2}
\indent For two random samples $X_{1},\ldots,X_{m}$ and $Y_{1},\ldots,Y_{n}$ in $\mathbb{R}$, the Wilcoxon--Mann--Whitney statistic is defined as $\sum_{i=1}^{m}\sum_{j=1}^{n} sign(Y_{j} - X_{i})$ (see, e.g., \cite{HSS99}). Since for any $x \neq 0$, $sign(x)$ is the derivative of $|x|$, we  define a notion of spatial rank for probability distributions in Banach spaces as follows. Let ${\bf X}$ be a random element in a Banach space ${\cal X}$, and ${\cal X}^{*}$ be the dual of ${\cal X}$, which is the Banach space of real-valued continuous linear functions on ${\cal X}$. Suppose that ${\cal X}$ is smooth, i.e., the norm $||.||$ in ${\cal X}$ is G\^ateaux differentiable (see, e.g., Section $2$, Chapter $4$ in \cite{BV10}) at each ${\bf x} \neq {\bf 0}$ with G\^ateaux derivative, say, $SGN_{{\bf x}} \in {\cal X}^{*}$. In other words, we assume that $\lim_{t \rightarrow 0} t^{-1}(||{\bf x} + t{\bf h}|| - ||{\bf x}||) = SGN_{{\bf x}}({\bf h})$ for all ${\bf x} \neq {\bf 0}$ and ${\bf h} \in {\cal X}$. In a Hilbert space ${\cal X}$, $SGN_{{\bf x}} = {\bf x}/||{\bf x}||$. If ${\cal X} = L_{p}[a,b]$ for some $p \in (1,\infty)$, which is the Banach space of all functions ${\bf x} : [a,b] \rightarrow \mathbb{R}$ satisfying $\int_{a}^{b} |{\bf x}(s)|^{p}ds < \infty$, then $SGN_{{\bf x}}({\bf h}) = \int_{a}^{b} sign\{{\bf x}(s)\}|{\bf x}(s)|^{p-1}{\bf h}(s)ds/||{\bf x}||^{p-1}$ for all ${\bf h} \in L_{p}[a,b]$. We define $SGN_{{\bf x}} = {\bf 0}$ if ${\bf x} = {\bf 0}$. The spatial rank of ${\bf x} \in {\cal X}$ with respect to the distribution of a random element ${\bf X} \in {\cal X}$ is defined as $S_{{\bf x}} = E(SGN_{{\bf x} - {\bf X}})$, where the expectation is in the Bochner sense (see, e.g., Section $2$, Chapter $3$ in \cite{AG80}). In a Hilbert space ${\cal X}$, $S_{{\bf x}} = E\{({\bf x}-{\bf X})/||{\bf x} - {\bf X}||\}$, and the spatial rank defined in this way has been studied in  $\mathbb{R}^{d}$ by \cite{Chau96}, \cite{CM97}, \cite{Oja10} and \cite{HM11}. \\
\indent Let ${\bf X}_{1},\ldots,{\bf X}_{m}$ and ${\bf Y}_{1},\ldots,{\bf Y}_{n}$ be independent observations from two probability measures $P$ and $Q$ on a smooth Banach space ${\cal X}$. If we assume that $P$ and $Q$ differ by a shift $\Delta \in {\cal X}$ in the location, our Wilcoxon--Mann--Whitney type statistic for testing the hypothesis $H_{0} : \Delta = {\bf 0}$ against $H_{1} : \Delta \neq {\bf 0}$ is defined as $T_{WMW} = (mn)^{-1} \sum_{i=1}^{m} \sum_{j=1}^{n} SGN_{{\bf Y}_{j} - {\bf X}_{i}}$. Note that $T_{WMW}$ is a Banach space valued U-statistic (see, e.g., \cite{Boro96}) and is an unbiased estimator of $E(SGN_{{\bf Y} - {\bf X}})$. If $H_{0}$ holds, we have $E(SGN_{{\bf Y} - {\bf X}}) = {\bf 0}$. So, we reject the null hypothesis for large values of $||T_{WMW}||$. It is straightforward to verify that for any $c \in \mathbb{R}$, ${\bf a} \in {\cal X}$ and a bijective linear isometry $B$ on ${\cal X}$, the hypotheses $H_{0}$, $H_{1}$, and the test statistic remain invariant under the transformation ${\bf X} \mapsto cB({\bf X}) + {\bf a}$ and ${\bf Y} \mapsto cB({\bf Y}) + {\bf a}$. \\
\indent We shall now study the asymptotic distribution of the statistic $T_{WMW}$. A Banach space ${\cal X}$ is said to be of type $2$ if there exists a constant $b > 0$ such that for any $n\geq 1$ and independent zero mean random elements ${\bf U}_{1},{\bf U}_{2},\ldots,{\bf U}_{n}$ in ${\cal X}$ satisfying $E(||{\bf U}_{i}||^{2}) < \infty$ for all $i=1,2,\ldots,n$, we have $E(||\sum_{i=1}^{n} {\bf U}_{i}||^{2}) \leq b\sum_{i=1}^{n} E(||{\bf U}_{i}||^{2})$ (see, e.g., Section $7$, Chapter $3$ in \cite{AG80}). Type $2$ Banach spaces are the only Banach spaces, where the central limit theorem holds for every sequence of independent and identically distributed random elements, whose squared norms have finite expectations. It is known that Hilbert spaces and the  $L_{p}$ spaces with $p \in [2,\infty)$ are type $2$ Banach spaces. We denote by $G({\bf m},{\bf C})$ the distribution of a Gaussian random element (say, ${\bf W}$) in a separable Banach space ${\cal X}$ with mean ${\bf m} \in {\cal X}$ and covariance ${\bf C}$, where ${\bf C} : {\cal X}^{*} \times {\cal X}^{*} \rightarrow \mathbb{R}$ is a symmetric nonnegative definite continuous bilinear functional. Note that for any ${\bf l} \in {\cal X}^{*}$, ${\bf l}({\bf W})$ has a Gaussian distribution on $\mathbb{R}$ with mean ${\bf l}({\bf m})$ and variance ${\bf C}({\bf l},{\bf l})$. Define $\mu = E(SGN_{{\bf Y} - {\bf X}})$. We denote by $\Gamma_{1}, \Gamma_{2} : {\cal X}^{**} \times {\cal X}^{**} \rightarrow \mathbb{R}$ the symmetric nonnegative definite continous bilinear functionals given by $\Gamma_{1}({\bf f},{\bf g}) = E[{\bf f}\{E(SGN_{{\bf Y} - {\bf X}}\mid{\bf X})\}{\bf g}\{E(SGN_{{\bf Y} - {\bf X}}\mid{\bf X})\}] - {\bf f}(\mu){\bf g}(\mu)$, and $\Gamma_{2}({\bf f},{\bf g}) = E[{\bf f}\{E(SGN_{{\bf Y} - {\bf X}}\mid{\bf Y})\}{\bf g}\{E(SGN_{{\bf Y} - {\bf X}}\mid{\bf Y})\}] - {\bf f}(\mu){\bf g}(\mu)$, where ${\bf f}, {\bf g} \in {\cal X}^{**}$. Note that for a random element ${\bf Z}$ in a Banach space with $E(||{\bf Z}||) < \infty$, the conditional 
expectation of ${\bf Z}$ given ${\bf X}$ exists and can be properly defined (see, e.g., Section $4$, Chapter II in \cite{VTC87} for the relevant details).
\begin{theorem} \label{thm1}
Let $N = m + n$ and $m/N \rightarrow \gamma \in (0,1)$ as $m, n \rightarrow \infty$. Also, assume that the dual space ${\cal X}^{*}$ is a separable and type $2$ Banach space. Then, for any two probability measures $P$ and $Q$ on ${\cal X}$, $(mn/N)^{1/2} (T_{WMW} - \mu)$ converges {\it weakly} to $G\{{\bf 0},(1-\gamma)\Gamma_{1} + \gamma\Gamma_{2}\}$ as $m,n \rightarrow \infty$.
\end{theorem}
\indent The implementation of the test can be done using the asymptotic distribution of $T_{WMW}$ under the null hypothesis. Under $H_{0}$, we have $\Gamma_{1} = \Gamma_{2}$. Let $c_{\alpha}$ denote the $100(1-\alpha)$th percentile of the distribution of $||G({\bf 0},\Gamma_{1})||$. Thus, our  test, which rejects $H_{0}$ if $||(mn/N)^{1/2}T_{WMW}|| > c_{\alpha}$ has asymptotic size $\alpha$. When ${\cal X}$ is a separable Hilbert space, $\Gamma_{1}$ has a spectral decomposition (see Theorem IV.$2.4$ in page $213$ and Proposition $1.9$ in page $161$ in \cite{VTC87}), which implies that $||G({\bf 0},\Gamma_{1})||^{2}$ is distributed as a weighted sum of independent chi-square random variables each with $1$ degree of freedom, and the weights are the eigenvalues of $\Gamma_{1}$. Further details about the implementation of our test are given below when we analyze some real data. It follows from Theorem \ref{thm1} that the asymptotic power of our test will be $1$ whenever $\mu \neq {\bf 0}$. This holds in particular if $Q$ differs from $P$ by a non-zero shift $\Delta$ in the location, ${\cal X}$ is a reflexive and strictly convex Banach space, and the distribution of ${\bf Y} - {\bf X}$ is nonatomic and not concentrated on a line in ${\cal X}$ (see, e.g., Theorem $4.14$ in \cite{Kemp87}). In other words, the test is consistent for location shift alternatives. \\
\indent We have applied our test based on $T_{WMW}$ to three real datasets, namely, the Coffee data, the Berkeley growth data and the Spectrometry data. The Coffee data is obtained from \texttt{http://www.cs.ucr.edu/$\sim$eamonn/time\_series\_data/} and contains spectroscopy readings taken at $286$ wavelength values for $14$ samples of each of the two varieties of coffee, namely, Arabica and Robusta. The Berkeley growth data is available in the R package ``fda'' (see \texttt{http://rss.acs.unt.edu/Rdoc/library/fda/html/growth.html}) and contains the heights of $39$ boys and $54$ girls measured at $31$ time points between the ages $1$ and $18$ years. The curves have been pre-smoothed using a monotone spline smoothing technique available in the R package ``fda''. The curves are recorded at $101$ equispaced ages in the interval $[1,18]$. The Spectrometry data is available at \texttt{http://www.math.univ-toulouse.fr/staph/npfda} and contains the spectrometric curves for $215$ meat units measured at $100$ wavelengths between $850$ nm and $1050$ nm. The data also contains the fat content of each meat unit, which is categorized into two classes, namely, ``$\leq 20\%$'' and ``$> 20\%$''. In all these three datasets, each observation can be viewed as an element in the separable Hilbert space $L_{2}[a,b]$. For instance, the spectrometric curves in the third dataset can be viewed as elements in the space $L_{2}[850,1050]$. \\
\indent In view of Theorem \ref{thm1} and the discussion following it, for all three real datasets, the asymptotic null distribution of $||(mn/N)^{1/2}T_{WMW}||$ can be expressed in terms of a weighted sum of independent chi-square random variables. Since only a few eigenvalues of the sample analog of $\Gamma_{1}$ are positive, we get a finite sum, and use its distribution, which can be simulated, to estimate the critical value of our test statistic. For each dataset, the norm in the definition of $SGN_{{\bf x}}$ used in $T_{WMW}$ is computed as the norm of the Euclidean space whose dimension is the number of time points over which the sample curves in that dataset are observed. We have also applied the two sample version of the test studied by \cite{CFF04} and the two tests of \cite{HKR13} to these datasets. We have used the usual empirical pooled covariance for the two tests of \cite{HKR13}, and the numbers of projection directions used in these two tests are chosen using the cumulative variance method 
described in their paper. For the Coffee data, the p-value of our test based on $T_{WMW}$ is $0.072$, that of the test in \cite{CFF04} is $0.169$, and both the tests in \cite{HKR13} have the same p-value $0.273$. None of the tests yield a very strong evidence against the null hypothesis, and all of them fail to reject it at the $5\%$ level. However, among the four p-values, the one obtained using our test bears the strongest evidence in favour of the alternative hypothesis. The p-values of all four tests for both of the Berkeley growth data and the Spectrometry data are $0$ upto two decimal places. We have also applied these tests to randomly chosen $20\%$ subsamples of the two datasets instead of the full datasets in order to investigate whether there is any difference in the results obtained using these tests when the sample sizes are smaller. The random subsampling was repeated $1000$ times for each dataset to compute the proportion of times each test rejects the null hypothesis when the level is fixed 
at $5\%$ for each test. For the subsamples of the Berkeley growth data and the Spectrometry data, the proportions of rejections of the null hypothesis by our test based on $T_{WMW}$ are $0.829$ and $0.832$, respectively, while those proportions are $0.476$ and $0.712$, respectively, for the test in \cite{CFF04}. The proportions of rejections of the null hypothesis by one of the two tests in \cite{HKR13} are $0.271$ and $0.744$ for the subsamples of the Berkeley growth data and the Spectrometry data, respectively, while those proportions are $0.292$ and $0.778$, respectively, for the other test in their paper. Thus, for the Berkeley growth data, our test has the highest rate of rejection of the null hypothesis, and those rates for the other three tests are much lower. For the Spectrometry data, all four tests have fairly high rates of rejection of the null hypothesis, and the rate is highest for our test using $T_{WMW}$.

\section{Asymptotic powers of different tests under shrinking location shifts}
\label{3}
\indent In the previous section, we have established the consistency of our test for models with fixed location shifts. We shall now derive the asymptotic distribution of our test statistic under appropriate sequences of shrinking location shifts. Suppose that ${\bf Y}$ is distributed as ${\bf X} + \Delta_{N}$, where $\Delta_{N} = \delta(mn/N)^{-1/2}$ for some fixed non-zero $\delta \in {\cal X}$ and $N \geq 1$. Recall that $N = m + n$ is the total size of the two samples. For some of the Wilcoxon--Mann--Whitney type tests studied in the finite dimensional setting, such alternative hypotheses have been shown to be contiguous to the null, and this leads to nondegenerate limiting distributions of the test statistics under those alternatives (see, e.g., \cite{CM97}, \cite{CC99} and \cite{Oja99}). For our next theorem, we assume that the norm in ${\cal X}$ is twice G\^ateaux differentiable at every ${\bf x} \neq {\bf 0}$ (see, e.g., Chapter $4$, Section $6$ in \cite{BV10}). Let us denote the Hessian of the function 
${\bf x} \mapsto E(||{\bf Y} - {\bf X} + {\bf x}||)$ at ${\bf x}$ by $J_{{\bf x}} : {\cal X} \rightarrow {\cal X}^{*}$ when it exists. In other words, for every ${\bf h} \in {\cal X}$,
\begin{eqnarray}
E(SGN_{{\bf Y} - {\bf X} + {\bf x} + t{\bf h}}) = E(SGN_{{\bf Y} - {\bf X} + {\bf x}}) + tJ_{{\bf x}}({\bf h}) + {\bf R}(t),   \label{eq3.1.1}
\end{eqnarray}
where $||{\bf R}(t)||/t \rightarrow 0$ as $t \rightarrow 0$. It is known that the norms in Hilbert spaces and the $L_{p}$ spaces with $p \in [2,\infty)$ are twice G\^ateaux differentiable. Let ${\cal X} = L_{p}[a,b]$ for some $2 \leq p < \infty$ and $-\infty < a < b < \infty$. If $E(||{\bf Y} - {\bf X} + {\bf x}||^{-1}) < \infty$, it can be shown that $J_{{\bf x}}$ exists and is given by
\begin{eqnarray*}
\{J_{{\bf x}}({\bf z})\}({\bf w}) &=& (p-1)E\left[\frac{\int_{a}^{b} |{\bf Y}(s)-{\bf X}(s)+{\bf x}(s)|^{p-2}{\bf z}(s){\bf w}(s)ds}{||{\bf Y} - {\bf X} + {\bf x}||^{p-1}} \right. \\
&& - \ \left.\frac{\left\{\int_{a}^{b} |{\bf Y}(s)-{\bf X}(s)+{\bf x}(s)|^{p-1}{\bf z}(s)ds\right\} \left\{\int_{a}^{b} |{\bf Y}(s)-{\bf X}(s)+{\bf x}(s)|^{p-1}{\bf w}(s)ds\right\}}{||{\bf Y} - {\bf X} + {\bf x}||^{2p-1}} \right],
\end{eqnarray*}
where ${\bf z}, {\bf w}$ and ${\bf x} \in L_{p}[a,b]$.
\begin{theorem}  \label{thm2}
As before, let $N = m + n$, $m/N \rightarrow \gamma \in (0,1)$ as $m, n \rightarrow \infty$, and ${\cal X}^{*}$ is a separable and type $2$ Banach space. Also, assume that the distribution of ${\bf X}$ is nonatomic and $J_{{\bf 0}}$ exists. Then, under the sequence of shrinking location shifts described at the beginning of this section, $(mn/N)^{1/2} T_{WMW}$ converges {\it weakly} to $G\{J_{{\bf 0}}(\delta),\Gamma_{1}\}$ as $m, n \rightarrow \infty$.
\end{theorem}
\indent In order to compare the asymptotic power of our test with those of the tests available in \cite{CFF04} and \cite{HKR13}, we shall now study the asymptotic distributions of those test statistics under the sequences of shrinking shifts described at the beginning of this section. For the two sample problem in $L_{2}[a,b]$, the test statistic studied by \cite{CFF04} reduces to $T_{CFF} = m||\bar{{\bf X}} - \bar{{\bf Y}}||^{2}$. \cite{HKR13} studied the test statistics $T_{HKR1} = \sum_{k=1}^{L} (\langle\bar{{\bf X}} - \bar{{\bf Y}},\widehat{\psi}_{k}\rangle)^{2}$ and $T_{HKR2} =  \sum_{k=1}^{L} \widehat{\lambda}_{k}^{-1} (\langle\bar{{\bf X}} - \bar{{\bf Y}},\widehat{\psi}_{k}\rangle)^{2}$. Here, $\langle.,.\rangle$ denotes the inner product in $L_{2}[0,1]$, the $\widehat{\lambda}_{k}$'s denote the eigenvalues of the empirical pooled covariance of the ${\bf X}_{i}$'s and the ${\bf Y}_{j}$'s in descending order of magnitudes, and the $\widehat{\psi}_{k}$'s are the corresponding empirical eigenvectors. If ${\cal X} = \mathbb{R}^{d}$ and $L = d$, $T_{HKR2}$ reduces to Hotelling's $T^{2}$ statistic, and $T_{HKR1} = m^{-1}T_{CFF}$. We derive the asymptotic distributions of $T_{HKR1}$ and $T_{HKR2}$ in a separable Hilbert space. Since the statistic $T_{CFF}$ can be defined in any Banach space, we derive its asymptotic distribution in a separable and type $2$ Banach space.
\begin{theorem}  \label{thm7}
Once again, let $N = m + n$ and $m/N \rightarrow \gamma \in (0,1)$ as $m, n \rightarrow \infty$. Then, under the sequence of shrinking location shifts mentioned at the beginning of this section, we have the following. \\
(a) If $E(||{\bf X}||^{2}) < \infty$, $nN^{-1}T_{CFF}$ converges {\it weakly} to $||G(\delta,\Sigma)||^{2}$ as $m,n \rightarrow \infty$, where $\Sigma$ denotes the covariance of ${\bf X}$. \\
(b) Assume that for some $L \geq 1$, $\lambda_{1} > \ldots > \lambda_{L} > \lambda_{L+1} > 0$, where the $\lambda_{k}$'s are the eigenvalues of $\Sigma$ in decreasing order of magnitudes. If $E(||{\bf X}||^{4}) < \infty$, $mnN^{-1}T_{HKR1}$ converges {\it weakly} to $\sum_{k=1}^{L} \lambda_{k}\chi^{2}_{(1)}(\beta_{k}^{2}/\lambda_{k})$, and $mnN^{-1}T_{HKR2}$ converges {\it weakly} to $\sum_{k=1}^{L} \chi^{2}_{(1)}(\beta_{k}^{2}/\lambda_{k})$ as $m, n \rightarrow \infty$. Here, $\beta_{k} = \langle\delta,\psi_{k}\rangle$, $\chi^{2}_{(1)}(\beta_{k}^{2}/\lambda_{k})$ denotes the noncentral chi-square variable with $1$ degree of freedom and noncentrality parameter $\beta_{k}^{2}/\lambda_{k}$, and $\psi_{k}$ is the eigenvector corresponding to $\lambda_{k}$ for $k = 1, 2, \ldots, L$.
\end{theorem}
\indent For evaluating the asymptotic powers of different tests under shrinking location shifts, we have considered some probability distributions in $L_{2}[0,1]$. Let ${\bf X} = \sum_{k=1}^{\infty} Z_{k}\phi_{k}$, where the $Z_{k}$'s are independent random variables, and $\phi_{k}(t) = \sqrt{2}sin\{(k-0.5){\pi}t\}$ for $k \geq 1$, which form an orthonormal basis of $L_{2}[0,1]$. We have considered two cases, namely, $Z_{k}/\sigma_{k}$ having a $N(0,1)$ distribution and a t distribution with $5$ degrees of freedom, where $\sigma_{k} = \{(k-0.5)\pi\}^{-1}$ for each $k \geq 1$. Both of these distributions satisfy the assumptions made in Theorems \ref{thm2} and \ref{thm7}. These two cases correspond to the Karhunen-Lo\`eve expansions of the standard Brownian motion (the sBm distribution) and the centered t process on $[0,1]$ with $5$ degrees of freedom (the t($5$) distribution) and covariance kernel $K(t,s) = \min(t,s)$ (see, e.g., \cite{YTY07}), respectively. Recall that ${\bf Y}$ is distributed as ${\bf X} + \delta(mn/N)^{-1/2}$, and we have considered three choices of $\delta$, namely, $\delta_{1}(t) = c$, $\delta_{2}(t) = ct$ and $\delta_{3}(t) = ct(1-t)$, where $t \in [0,1]$ and $c > 0$. For evaluating the asymptotic powers of different tests using Theorems \ref{thm2} and \ref{thm7}, the expectations appearing in $J_{{\bf 0}}(\delta)$ and $\Gamma_{1}$ can be evaluated numerically using averages over Monte-Carlo replications of relevant random objects. For an appropriately large $d$, the eigenvalues and the eigenvectors of $\Gamma_{1}$ and $\Sigma$ can be approximated by those of the $d \times d$ covariance matrices associated with the values of the sample functions at $d$ equispaced points in $[0,1]$. Figure \ref{Fig:3} shows the plots of the ratios of the asymptotic powers of the tests based on $T_{CFF}$, $T_{HKR1}$ and $T_{HKR2}$ to those of our test using $T_{WMW}$. It is seen that all the tests attain the $5\%$ nominal level asymptotically for both of the distributions, and the curves in each plot in Figure \ref{Fig:3} meet at $c = 0$. The asymptotic powers of the tests based on $T_{CFF}$ and $T_{HKR1}$ are close for all the situations considered. The test using $T_{HKR2}$ outperforms both of them for $\delta_{1}(t)$ and $\delta_{3}(t)$, but is asymptotically less powerful for $\delta_{2}(t)$ for both the distributions considered. The asymptotic powers of the tests based on $T_{CFF}$ and $T_{HKR1}$ are close to that of our test using $T_{WMW}$ for $\delta_{2}(t)$ under the sBm distribution, and in all other cases, our test outperforms them. The test based on $T_{HKR2}$ outperforms our test for $\delta_{3}(t)$ for both the distributions, while it is asymptotically more powerful for small $c$ values for $\delta_{1}(t)$ under the sBm distribution. In other cases, our test outperforms this test.
\begin{figure}[ht]
\centering
 \includegraphics[height=4in,width=5in]{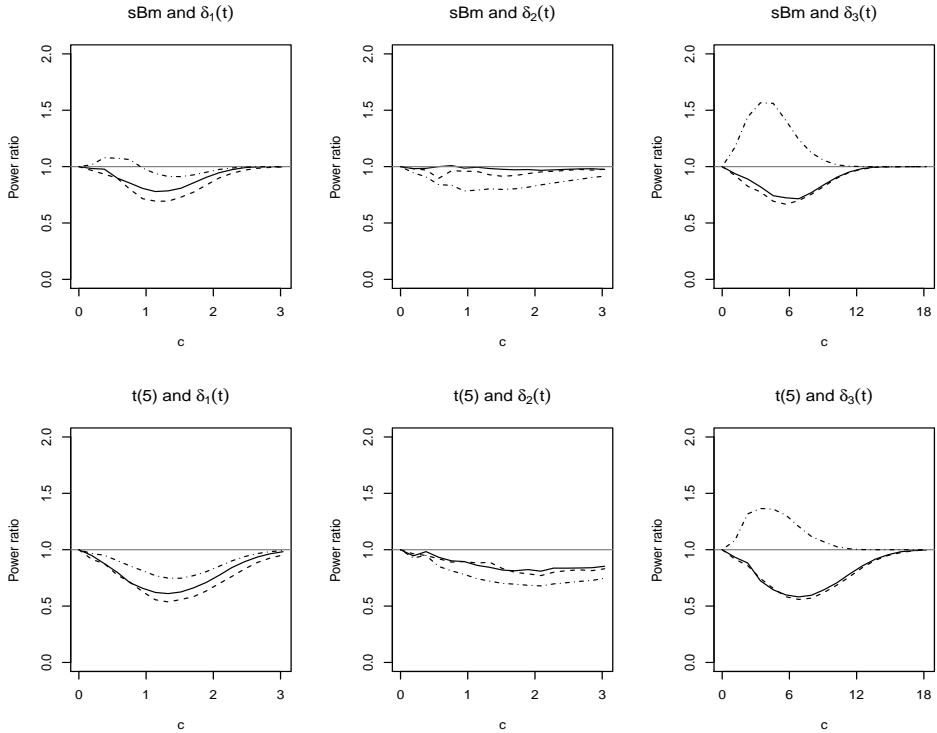}
\caption{Plots of the ratios of the asymptotic powers of the tests based on $T_{CFF}$ ({\it solid line}), $T_{HKR1}$ ({\it dashed line}) and $T_{HKR2}$ ({\it dot-dashed line}) to those of our test using $T_{WMW}$ for the sBm and the t($5$) distributions under shrinking location shifts.}
\label{Fig:3}
\end{figure}

\section{Finite sample powers of different tests}
\label{4.1}
\indent In this section, we shall carry out a comparative study of the finite sample empirical powers of the tests considered in Sections \ref{2} and \ref{3} and a few other tests. Once again, let ${\bf X} = \sum_{k=1}^{\infty} Z_{k}\phi_{k}$, where the $Z_{k}$'s and the $\phi_{k}$'s are as in Section \ref{3}, and the distributions of ${\bf Y}$ and ${\bf X}$ differ by a shift $\Delta$. Here, we have considered three cases, namely, $Z_{k}/\sigma_{k}$ having a $N(0,1)$ distribution (the sBm distribution) and $Z_{k}/\sigma_{k} = U_{k}(V/r)^{-1/2}$, where $U_{k}$'s are independent $N(0,1)$ variables and $V$ has a chi-square distribution with $r$ degrees of freedom for $r = 1$ and $5$ independent of the $U_{k}$'s for each $k \geq 1$ (the t($1$) and the t($5$) distributions, respectively), where the $\sigma_{k}$'s are as in Section \ref{3}. The t($1$) distribution is included to investigate the performance of our test and its competitors when the moment conditions on the probability distribution required by the tests based on $T_{CFF}$, $T_{HKR1}$ and $T_{HKR2}$ (see \cite{CFF04} and \cite{HKR13}) fail to hold. Note that the conditions assumed for our test (see Theorem \ref{thm1}) hold for all the distributions considered here. We have chosen $m = n = 15$, and each sample curve is observed at $250$ equispaced points in $[0,1]$. Three types of shifts are considered, namely, $\Delta_{1}(t) = c$, $\Delta_{2}(t) = ct$ and $\Delta_{3}(t) = ct(1-t)$, where $t \in [0,1]$ and $c > 0$ (cf. Section \ref{3}). For each simulated dataset, all the test statistics and their critical values are computed in the same way as in Section \ref{2}, where we analyzed some real datasets. All the sizes and the powers are evaluated by averaging the results of $1000$ Monte-Carlo simulations. Figure \ref{Fig:1} shows the plots of the ratios of the finite sample powers of the competing tests to those of our test at the nominal level of $5\%$. \\
\indent The sizes of all the tests considered in Sections \ref{2} and \ref{3} are close to the nominal $5\%$ level for the sBm and the t($5$) distributions. For the t($1$) distribution, all those tests have sizes around $1.5\%$, while our test using $T_{WMW}$ has size $4.4\%$. The test using $T_{HKR2}$ outperforms the tests based on $T_{CFF}$ and $T_{HKR1}$ in all the situations considered except for $\Delta_{2}(t)$ under the sBm and the t($5$) distributions, where it is less powerful for larger $c$ values. The tests based on $T_{CFF}$ and $T_{HKR1}$ have similar powers for all the models considered. Their powers coincide for all the shifts under the t($1$) distribution, where our test outperforms all three competitors. For all the shifts under the t($5$) distribution, our test outperforms both the tests using $T_{CFF}$ and $T_{HKR1}$. For $\Delta_{1}(t)$ and $\Delta_{2}(t)$ under the t($5$) distribution, our test is more powerful than the test using $T_{HKR2}$ except for small values of $c$, and this latter test 
outperforms our test for $\Delta_{3}(t)$. The behaviour of all the tests under the sBm distribution is similar to that under the t($5$) distribution, except for $\Delta_{2}(t)$. For $\Delta_{2}(t)$ under the sBm distribution, the three competing tests have an edge over our test. These finite sample results are broadly in conformity with the asymptotic results in Section \ref{3}.
\begin{figure}
\centering
 \includegraphics[height=6.4in,width=5.8in]{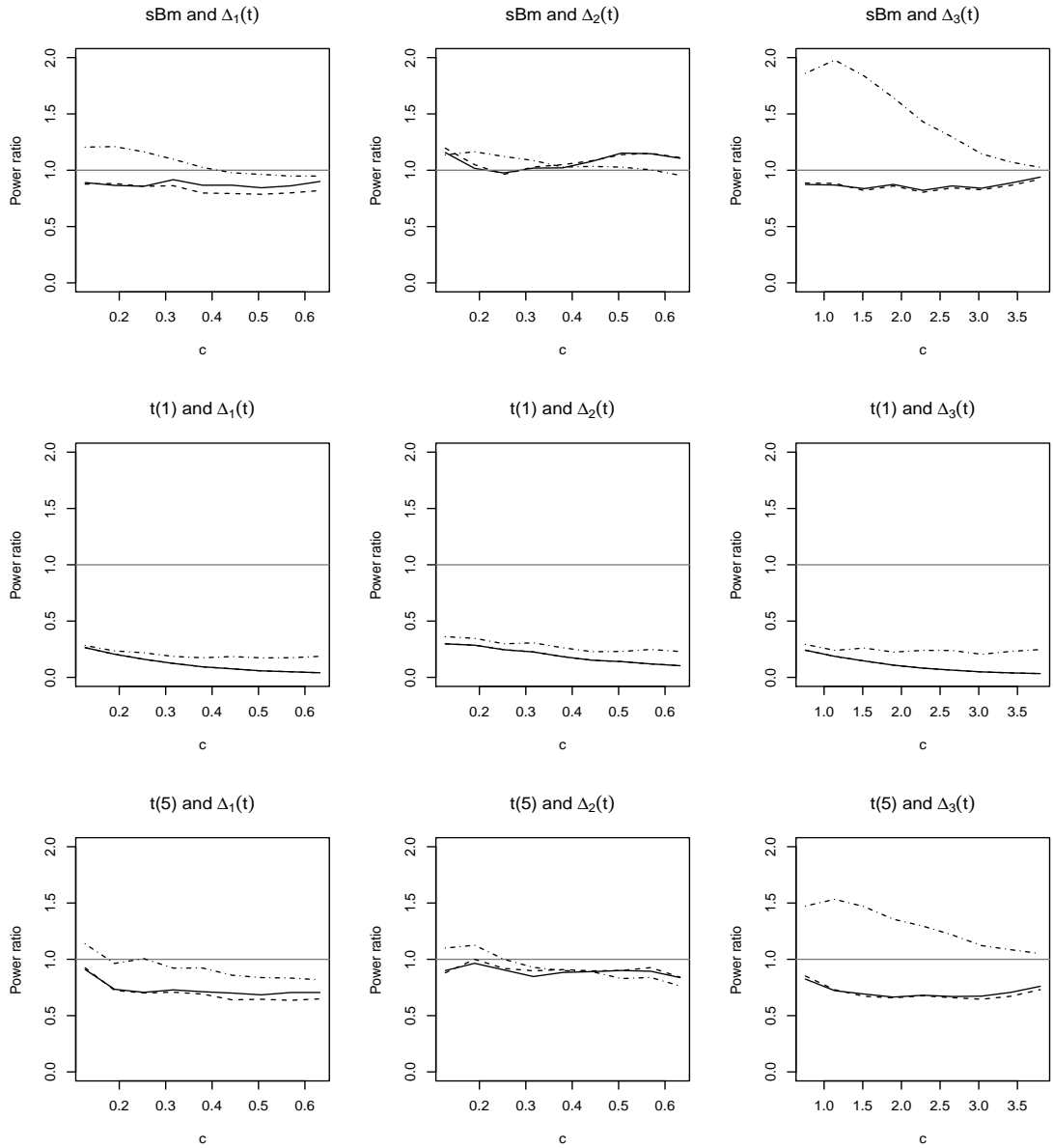}
\caption{Plots of the ratios of the finite sample powers of the tests based on $T_{CFF}$ ({\it solid line}), $T_{HKR1}$ ({\it dashed line}) and $T_{HKR2}$ ({\it dot-dashed line}) to those of our test using $T_{WMW}$ for the sBm, the t($1$) and the t($5$) distributions under location shifts.}
\label{Fig:1}
\end{figure}
\indent We have compared the finite sample powers of our test and some more tests available in the literature. A pointwise t-test with an appropriate p-value correction for multiple testing was studied by \cite{CL08} for testing the equality of means of two Gaussian functional datasets. \cite{SF04} studied some $F$ test for linear models involving Gaussian functional data, and we consider the two sample version of this test. \cite{CAFB10} studied an analysis of variance test for functional data based on multiple testing using random univariate linear projections of the data. The two sample version of their test reduces to the Wilcoxon--Mann--Whitney test based on such projections of the data. \cite{GBRSS12} studied a test for comparing two probability distributions on metric spaces, which may not necessarily differ by a shift in the location. We have used their test based on the asymptotic distribution of the unbiased statistic MMD$_{u}^{2}$ (see Section $5$ in \cite{GBRSS12}). For comparing two finite 
dimensional probability distributions, \cite{HT02} studied a permutation test based on the ranks of the distances between the sample observations, while \cite{Rose05} studied a test based on a notion of adjacency. The authors of both papers pointed out that these tests can be used for infinite dimensional functional data as well. The asymptotic behaviours of none of the above-mentioned tests under the type of shrinking shifts considered in Section \ref{3} are known in the literature, and such an analysis is beyond the scope of this paper. We only carry out a finite sample empirical power comparison of our test based on $T_{WMW}$ with these tests. We have used the $L_{2}$-distance between the pointwise ranks as the distance function for implementing the Rosenbaum test. The discrete distribution of the test statistic made the size of this test much less than the nominal significance level for the small sample sizes that we have considered. To rectify this, we considered a randomized version of the test, and this improved the size as well as the power of the test. We have chosen $30$ random projections for implementing the Cuesta-Albertos and Febrero-Bande test, as recommended by these authors. We have chosen the radial basis function as the kernel for the Gretton et al. test and used the codes provided by these authors. All the other tests are implemented using our own codes, and all the sizes and the powers are evaluated by averaging the results of $1000$ Monte-Carlo simulations. Figure \ref{Fig:2} shows the plots of the ratios of the powers of these tests to those of our test using $T_{WMW}$. \\
\indent For all the distributions considered, the sizes of the Rosenbaum test, the Hall--Tajvidi test, the Gretton et al. test and the Cox--Lee test were close to the nominal $5\%$ level. The sizes of the Cuesta-Albertos and Febrero-Bande test were around $2.6\%$ in all our simulations. The sizes of the Shen--Faraway test were much smaller than the nominal level for all the distributions considered, and it was zero for the t($1$) distribution. Figure \ref{Fig:2} shows that our test based on $T_{WMW}$ is uniformly more powerful than the Cuesta-Albertos and Febrero-Bande test and the Shen--Faraway test in all the situations considered. Our test outperforms the Rosenbaum test, the Hall--Tajvidi test and the Gretton et al. test in all but the following situations. The Rosenbaum test and the Hall--Tajvidi test are more powerful than our test for small values of $c$ for $\Delta_{2}(t)$ under all the distributions. The Hall--Tajvidi test is also more powerful than our test for small $c$ values for $\Delta_{1}(t)$ and $\
Delta_{3}(t)$ under the t($1$) distribution. The Gretton et al. test has a slight edge over our test for all the shifts under the t($1$) distribution. \\
\indent Except for small $c$ values, our test using $T_{WMW}$ is more powerful than the Cox--Lee test for $\Delta_{2}(t)$ under all the distributions and for the shift $\Delta_{3}(t)$ under the t($1$) distribution. For $\Delta_{3}(t)$ under the sBm and the t($5$) distributions, the Cox--Lee test has an edge over our test. For $\Delta_{1}(t)$, the Cox--Lee test is far more superior to our test for all of the three distributions considered, and we have not plotted the ratios of its power to those of our test, since the values lie beyond the plotting ranges used in Figure \ref{Fig:2}. The reason for such a behaviour of this test is that the coordinate random variable at $t = 0.0001$ (which is closest to zero in our computations) has scale parameter equal to $0.0001$ for all the distributions considered. Consequently, for this coordinate and $\Delta_{1}(t)$, the adjusted p-values of the t-test used in the Cox--Lee procedure are $\leq 0.05$ for many of the simulations. The Cox--Lee test rejects $H_{0}$ for such 
simulations resulting in the high power of this test for this shift.
\begin{figure}
\centering
 \includegraphics[height=6.4in,width=5.8in]{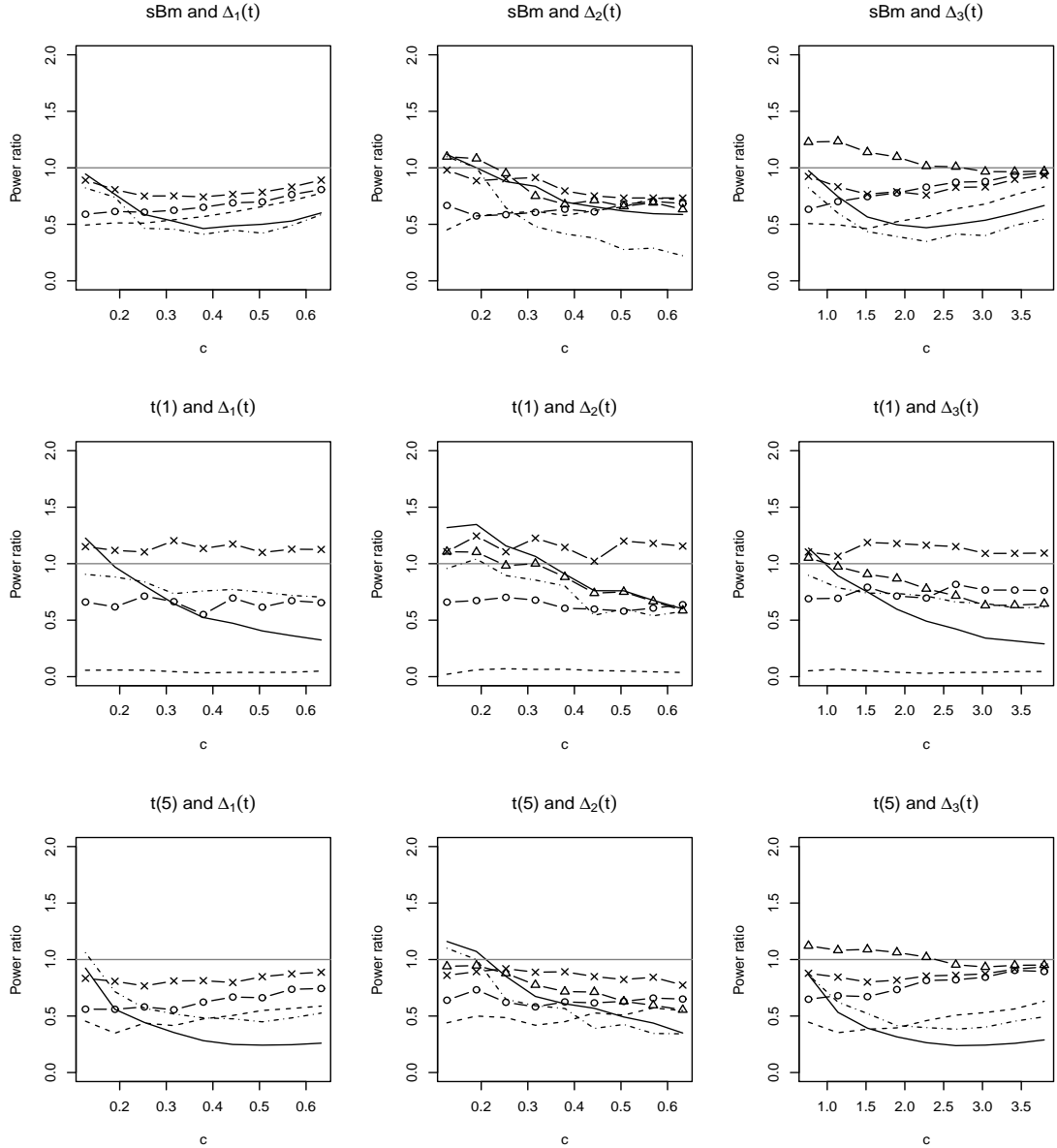}
\caption{Plots of the ratios of the finite sample powers of the Hall--Tajvidi test ({\it solid line}), the Shen--Faraway test ({\it dashed line}), the Rosenbaum test ({\it dot-dashed line}), the Cox--Lee test (--$\vartriangle$--), the Cuesta-Albertos and Febrero-Bande test  (--$\circ$--) and the Gretton et al. test (--$\times$--) to those of our test using $T_{WMW}$ for the sBm, the t($1$) and the t($5$) distributions under location shifts.}
\label{Fig:2}
\end{figure}

\section*{Acknowledgement}
Research of the first author is partially supported by the SPM Fellowship of the Council of Scientific and Industrial Research, Government of India.

\section*{Appendix -- Proofs of the theorems}
\label{7}
\begin{proof}[Proof of Theorem \ref{thm1}]
Observe that $T_{WMW} - \mu$ is a two-sample Banach space valued U-statistic with kernel ${\bf h}({\bf x},{\bf y}) = SGN_{{\bf y} - {\bf x}} - \mu$ satisfying $E\{{\bf h}({\bf X},{\bf Y})\} = {\bf 0}$. By the Hoeffding decomposition for Banach space valued U-statistics (see, e.g., Section $1.2$ in \cite{Boro96}), we have
\begin{eqnarray*}
 T_{WMW} - \mu = \frac{1}{m} \sum_{i=1}^{m} E\{{\bf h}({\bf X}_{i},{\bf Y})\mid{\bf X}_{i}\} + \frac{1}{n} \sum_{j=1}^{n} E\{{\bf h}({\bf X},{\bf Y}_{j})\mid{\bf Y}_{j}\} + {\bf R}_{m,n}.
\end{eqnarray*}
So, ${\bf R}_{m,n} = (mn)^{-1} \sum_{i=1}^{m}\sum_{j=1}^{n} \widetilde{{\bf h}}({\bf X}_{i},{\bf Y}_{j})$, where $\widetilde{{\bf h}}({\bf x},{\bf y}) = {\bf h}({\bf x},{\bf y}) - E\{{\bf h}({\bf X},{\bf Y})\mid{\bf X}={\bf x}\} - E\{{\bf h}({\bf X},{\bf Y})\mid{\bf Y}={\bf y}\}$. Let $\Phi({\bf X}_{i}) = \sum_{j=1}^{n} \widetilde{{\bf h}}({\bf X}_{i},{\bf Y}_{j})$. Since $E\{\widetilde{{\bf h}}({\bf X},{\bf Y})\mid{\bf Y}={\bf y}\} = {\bf 0}$ for all ${\bf y} \in {\cal X}$, using the definition of type $2$ Banach spaces mentioned in Section \ref{2}, we get
\begin{eqnarray}
 E(||{\bf R}_{m,n}||^{2}\mid{\bf Y}_{j};j=1,2,\ldots,n) &=& \frac{1}{m^{2}n^{2}} E\left\{\left\|\sum_{i=1}^{m} \Phi({\bf X}_{i})\right\|^{2}\mid{\bf Y}_{j};j=1,\ldots,n\right\}  \nonumber \\
 &\leq& \frac{b}{m^{2}n^{2}} \sum_{i=1}^{m} E\left\{||\Phi({\bf X}_{i})||^{2}\mid{\bf Y}_{j};j=1,\ldots,n\right\}. \label{eq1.0}
\end{eqnarray}
Taking expectations of both sides of \eqref{eq1.0} with respect to ${\bf Y}_{j}$ for $1 \leq j \leq n$, and using the fact that the ${\bf X}_{i}$'s are identically distributed, we get
\begin{eqnarray}
E(||{\bf R}_{m,n}||^{2}) &\leq& \frac{b}{mn^{2}} E\left\{\left\|\sum_{j=1}^{n} \widetilde{{\bf h}}({\bf X}_{1},{\bf Y}_{j})\right\|^{2}\right\}.  \label{eq1.1}
\end{eqnarray}
Since $E\{\widetilde{{\bf h}}({\bf X},{\bf Y})\mid{\bf X}={\bf x}\} = {\bf 0}$ for all ${\bf x} \in {\cal X}$, once again from the definition of type $2$ Banach spaces and the fact that the ${\bf Y}_{j}$'s are identically distributed, we get
\begin{eqnarray}
E\left\{\left\|\sum_{j=1}^{n} \widetilde{{\bf h}}({\bf X}_{1},{\bf Y}_{j})\right\|^{2}\right\} &=& E\left[E\left\{\left\|\sum_{j=1}^{n} \widetilde{{\bf h}}({\bf X}_{1},{\bf Y}_{j})\right\|^{2}\mid{\bf X}_{1}\right\}\right]  \nonumber \\
&\leq& bE\left[\sum_{j=1}^{n} E\left\{\left\|\widetilde{{\bf h}}({\bf X}_{1},{\bf Y}_{j})\right\|^{2}\mid{\bf X}_{1}\right\}\right] \nonumber \\
&=& bnE\left\{\left\|\widetilde{{\bf h}}({\bf X}_{1},{\bf Y}_{1})\right\|^{2}\right\}.    \label{eq1.2}
\end{eqnarray}	
Since $||SGN_{{\bf x}}|| \leq 1$ for all ${\bf x} \in {\cal X}$, we have $||\widetilde{{\bf h}}({\bf x},{\bf y})|| \leq 4$ for all ${\bf x}, {\bf y} \in {\cal X}$. Combining this fact with \eqref{eq1.1} and \eqref{eq1.2}, we have
\begin{eqnarray*}
E(||{\bf R}_{m,n}||^{2}) &\leq& \frac{b^{2}}{mn} E\left\{\left\|\widetilde{{\bf h}}({\bf X}_{1},{\bf Y}_{1})\right\|^{2}\right\} \ \leq \ \frac{16b^{2}}{mn}.
\end{eqnarray*}
Since $m/N \rightarrow \gamma \in (0,1)$, we get $E\{||(mn/N)^{1/2}{\bf R}_{m,n}||^{2}\} \rightarrow 0$ as $m,n \rightarrow \infty$. Hence,  $(mn/N)^{1/2}{\bf R}_{m,n}$ converges to ${\bf 0}$ {\it in probability} as $m,n \rightarrow \infty$. We note here that a similar result has been proved in \cite{Boro96} for Banach space valued U-statistics, but the proof given above is simpler and uses the fact that ${\bf h}$ is a bounded kernel. Such a result for real-valued U-statistics is discussed in Chapter $5$ in \cite{Serf80} under the assumption that the kernel has a finite second moment. Now, $m^{-1/2} \sum_{i=1}^{m} E\{{\bf h}({\bf X}_{i},{\bf Y})\mid{\bf X}_{i}\}$ and $n^{-1/2} \sum_{j=1}^{n} E\{{\bf h}({\bf X},{\bf Y}_{j})\mid{\bf Y}_{j}\}$ converge {\it weakly} to $G({\bf 0},\Gamma_{1})$ and $G({\bf 0},\Gamma_{2})$, respectively, as $m, n \rightarrow \infty$ by the central limit theorem for independent and identically distributed random variables in a separable and type $2$ Banach space (see, e.g, Theorem $7.5(i)$ in \cite{AG80}). So, the independence of these two sums, the assumption that $m/N \rightarrow \gamma$, and the fact that $(mn/N)^{1/2}{\bf R}_{m,n}$ converges to ${\bf 0}$ {\it in probability} complete the proof.
\end{proof}

\begin{proof}[Proof of Theorem \ref{thm2}]
Define $\rho(\Delta_{N}) = E(SGN_{{\bf Y} - {\bf X}})$. Applying the Hoeffding decomposition for Banach space valued U-statistics as in the proof of Theorem \ref{thm1}, it follows that
\begin{eqnarray}
 T_{WMW} - \rho(\Delta_{N}) &=& \frac{1}{m} \sum_{i=1}^{m} \{E(SGN_{{\bf Y} - {\bf X}_{i}}\mid{\bf X}_{i}) - \rho(\Delta_{N})\} \nonumber \\
 && + \ \frac{1}{n} \sum_{j=1}^{n} \{E(SGN_{{\bf Y}_{j} - {\bf X}}\mid{\bf Y}_{j}) - \rho(\Delta_{N})\} + {\bf S}_{m,n}.  \label{eq2.1}
\end{eqnarray}
Arguing as in the proof of Theorem \ref{thm1}, it can be shown that $E(||{\bf S}_{m,n}||^{2}) \leq 16b^{2}/mn$ for each $m, n \geq 1$. Thus, $(mn/N)^{1/2}{\bf S}_{m,n} \rightarrow {\bf 0}$ {\it in probability} as $m,n \rightarrow \infty$ under the sequence of shrinking shifts. \\
\indent Note that $\rho(\Delta_{N}) = E(SGN_{{\bf Z} - {\bf X} + \Delta_{N}})$, where ${\bf Z}$ is an independent copy of ${\bf X}$. So, it follows from \eqref{eq3.1.1} in Section \ref{3} that 
\begin{eqnarray}
(mn/N)^{1/2}\rho(\Delta_{N}) \longrightarrow J_{{\bf 0}}(\delta)  \label{eq2.2}
\end{eqnarray}
as $m, n \rightarrow \infty$. \\
\indent We next show the asymptotic Gaussianity of the first term on the right hand side of \eqref{eq2.1} after it is multiplied by $m^{1/2}$. Let us write $\Psi_{N}({\bf X}_{i}) = m^{-1/2}\{E(SGN_{{\bf Y} - {\bf X}_{i}}\mid{\bf X}_{i}) - \rho(\Delta_{N})\}$. Note that $E\{\Psi_{N}({\bf X}_{i})\} = {\bf 0}$. In order to show the asymptotic Gaussianity of $\sum_{i=1}^{m} \Psi_{N}({\bf X}_{i})$, it is enough to show that the triangular array $\{\Psi_{N}({\bf X}_{1}),\ldots,\Psi_{N}({\bf X}_{m})\}_{m=1}^{\infty}$ of rowwise independent and identically distributed random elements satisfy the conditions of Corollary $7.8$ in \cite{AG80}. \\
\indent Observe that for any $\epsilon > 0$,
\begin{eqnarray*}
 \sum_{i=1}^{m} P(||\Psi_{N}({\bf X}_{i})|| > \epsilon) \leq \sum_{i=1}^{m} E\{||E(SGN_{{\bf Y} - {\bf X}_{i}}\mid{\bf X}_{i}) - \rho(\Delta_{N})||^{3}\}/m^{3/2} \leq 8m^{-1/2}.
\end{eqnarray*}
Thus, $\lim_{m \rightarrow \infty} \sum_{i=1}^{m} P(||\Psi_{N}({\bf X}_{i})|| > \epsilon) = 0$ for every $\epsilon > 0$, which ensures that condition ($1$) of Corollary $7.8$ in \cite{AG80} holds. \\
\indent We next verify condition ($2$) of Corollary $7.8$ in \cite{AG80}. Let us fix ${\bf f} \in {\cal X}^{**}$. Since $||SGN_{{\bf x}}|| = 1$ for all ${\bf x} \neq {\bf 0}$, we can choose $\delta = 1$ in that condition ($2$). Then, using the linearity of ${\bf f}$, we have
\begin{eqnarray}
\sum_{i=1}^{m} E[{\bf f}^{2}\{\Psi_{N}({\bf X}_{i})\}] = m^{-1} \sum_{i=1}^{m} E[\{W_{N,i} - E(W_{N,i})\}^{2}],   \label{eq2.3}
\end{eqnarray}
where $W_{N,i} = {\bf f}\{E(SGN_{{\bf Y} - {\bf X}_{i}}\mid{\bf X}_{i})\}$. Since the ${\bf X}_{i}$'s are identically distributed, the right hand side in \eqref{eq2.3} simplifies to $E[\{W_{N,1} - E(W_{N,1})\}^{2}]$. Note that $W_{N,1} = {\bf f}\{E(SGN_{{\bf Z} - {\bf X}_{1} + \Delta_{N}}\mid{\bf X}_{1})\}$, where ${\bf Z}$ is an independent copy of ${\bf X}_{1}$. Since the norm in ${\cal X}$ is assumed to be twice G\^ateaux differentiable, it follows from Theorem $4.6.15$(a) and Proposition $4.6.16$ in \cite{BV10} that the norm in ${\cal X}$ is Fr\'echet differentiable. This in turn implies that the map ${\bf x} \mapsto SGN_{{\bf x}}$ is continuous on ${\cal X}\backslash \{{\bf 0}\}$ (see, e.g., Corollary $4.2.12$ in \cite{BV10}). Using this fact, it can be shown that 
\begin{eqnarray}
E(SGN_{{\bf Z} - {\bf X}_{1} + \Delta_{N}}\mid{\bf X}_{1}) \longrightarrow E(SGN_{{\bf Z} - {\bf X}_{1}}\mid{\bf X}_{1})    \label{eq2.4}
\end{eqnarray}
as $m, n \rightarrow \infty$ for {\it almost all} values of ${\bf X}_{1}$. Thus, we get the convergence of $E(W_{N,1})$ to $E[{\bf f}\{E(SGN_{{\bf Z} - {\bf X}_{1}}\mid{\bf X}_{1})\}]$ as $m,n \rightarrow \infty$. Similarly, it follows that $E(W_{N,1}^{2})$ converges to $E[{\bf f}^{2}\{E(SGN_{{\bf Z} - {\bf X}_{1}}\mid{\bf X}_{1})\}]$ as $m,n \rightarrow \infty$. So, $\sum_{i=1}^{m} E[{\bf f}^{2}\{\Psi_{N}({\bf X}_{i})\}] \rightarrow \Gamma_{1}({\bf f},{\bf f})$ as $m, n \rightarrow \infty$, where $\Gamma_{1}$ is as defined before Theorem \ref{thm1} in Section \ref{2}. This completes the verification of condition ($2$) of Corollary $7.8$ in \cite{AG80}. \\
\indent Finally, for the verification of condition ($3$) of Corollary $7.8$ in \cite{AG80}, suppose that $\{{\cal F}_{k}\}_{k \geq 1}$ is a sequence of finite dimensional subspaces of ${\cal X}^{*}$ such that ${\cal F}_{k} \subseteq {\cal F}_{k+1}$ for all $k \geq 1$, and the closure of $\bigcup_{k=1}^{\infty} {\cal F}_{k}$ is ${\cal X}^{*}$. Such a sequence of subspaces exists because of the separability of ${\cal X}^{*}$. For any ${\bf x} \in {\cal X}^{*}$ and any $k \geq 1$, we define $d({\bf x},{\cal F}_{k}) = \inf\{||{\bf x} - {\bf y}|| : {\bf y} \in {\cal F}_{k}\}$. It is straightforward to verify that for every $k \geq 1$, the map ${\bf x} \mapsto d({\bf x}, {\cal F}_{k})$ is continuous and bounded on any closed ball in ${\cal X}^{*}$. Thus, using \eqref{eq2.4}, it follows that $\rho(\Delta_{N}) \rightarrow 0$ as $m, n \rightarrow \infty$, and we have
\begin{eqnarray*}
 \sum_{i=1}^{m} E[d^{2}\{\Psi_{N}({\bf X}_{i}), {\cal F}_{k}\}] &=& m^{-1} \sum_{i=1}^{m} E[d^{2}\{E(SGN_{{\bf Z} - {\bf X}_{i} + \Delta_{N}}\mid{\bf X}_{i}) - \rho(\Delta_{N}), {\cal F}_{k}\}]  \\
&=& E[d^{2}\{E(SGN_{{\bf Z} - {\bf X}_{1} + \Delta_{N}}\mid{\bf X}_{1}) - \rho(\Delta_{N}), {\cal F}_{k}\}]  \\
&\longrightarrow& E[d^{2}\{E(SGN_{{\bf Z} - {\bf X}_{1}}\mid{\bf X}_{1}), {\cal F}_{k}\}]
\end{eqnarray*}
as $m, n \rightarrow \infty$. From the choice of the ${\cal F}_{k}$'s, it can be shown that $d({\bf x}, {\cal F}_{k}) \rightarrow 0$ as $k \rightarrow \infty$ for all ${\bf x} \in {\cal X}^{*}$. So, we have
\begin{eqnarray*}
\lim_{k \rightarrow \infty} E[d^{2}\{E(SGN_{{\bf Z} - {\bf X}_{1}}\mid{\bf X}_{1}), {\cal F}_{k}\}] = 0,
\end{eqnarray*}
and this completes the verification of condition ($3$) of Corollary $7.8$ in \cite{AG80}. \\
\indent Thus, $\sum_{i=1}^{m} \Psi_{N}({\bf X}_{i})$ converges {\it weakly} to a centered Gaussian random element in ${\cal X}^{*}$ as $m, n \rightarrow \infty$. Further, its asymptotic covariance is $\Gamma_{1}$, which was obtained while checking condition ($2$) of Corollary $7.8$ in \cite{AG80}. It follows from similar arguments that when the second term on the right hand side of \eqref{eq2.1} is multiplied by $n^{1/2}$, it also converges {\it weakly} to a Gaussian random element in ${\cal X}^{*}$ with the same distribution as $m, n \rightarrow \infty$. Hence, using the independence of the first two terms on the right hand side of \eqref{eq2.1}, we have
\begin{eqnarray*}
 (mn/N)^{1/2}\{T_{WMW} - \rho(\Delta_{N})\} \longrightarrow G({\bf 0},\Gamma_{1})
\end{eqnarray*}
{\it weakly} as $m,n \rightarrow \infty$ under the sequence of shrinking shifts. This, together with \eqref{eq2.2}, completes the proof of the theorem.
\end{proof}

\begin{proof}[Proof of Theorem \ref{thm7}]
(a) Let us observe that $nN^{-1}T_{CFF} = mnN^{-1}||\bar{{\bf X}} - \bar{{\bf Y}}||^{2}$. For each $N \geq 1$, ${\bf Y}$ has the same distribution as that of ${\bf Z} + \Delta_{N}$, where ${\bf Z}$ is an independent copy of ${\bf X}$. Now, by the central limit theorem for independent and identically distributed random elements in a separable and type $2$ Banach space (see, e.g. Theorem $7.5$(i) in \cite{AG80}), it follows that $(mn/N)^{1/2}(\bar{{\bf Z}} - \bar{{\bf X}})$ converges {\it weakly} to $G({\bf 0},\Sigma)$ as $m, n \rightarrow \infty$. Thus, $(mn/N)^{1/2}\left(\bar{{\bf Y}} - \bar{{\bf X}}\right)$, which has the same distribution as that of $(mn/N)^{1/2}\left(\bar{{\bf Z}} - \bar{{\bf X}} + \Delta_{N}\right)$, converges {\it weakly} to $G(\delta,\Sigma)$ as $m, n \rightarrow \infty$. This proves part (a) of the proposition. \\
(b) Let ${\bf v} = (\langle\bar{{\bf X}} - \bar{{\bf Y}},\psi_{1}\rangle,\ldots,\langle\bar{{\bf X}} - \bar{{\bf Y}},\psi_{L}\rangle)^{T}$ and $\boldsymbol\beta = (\beta_{1},\ldots,\beta_{L})^{T}$. It follows from the central limit theorem in $\mathbb{R}^{L}$ that $(mn/N)^{1/2}\{{\bf v} - (mn/N)^{-1/2}\boldsymbol\beta\}$ converges {\it weakly} to $N_{L}({\bf 0},\Lambda_{L})$ as $m, n \rightarrow \infty$ under the given sequence of shrinking shifts, where $\Lambda_{L}$ is the diagonal matrix $Diag(\lambda_{1},\ldots,\lambda_{L})$. Thus, under the given sequence of shifts, $(mn/N)^{1/2}{\bf v}$ converges {\it weakly} to a $N_{L}(\boldsymbol\beta,\Lambda_{L})$ distribution as $m, n \rightarrow \infty$. \\
\indent From arguments similar to those in the proof of Theorem 5.3 in \cite{HK12}, and using the assumptions in the present theorem, we get
\begin{eqnarray}
&& \max_{1 \leq k \leq L} (mn/N)^{1/2} \left|\langle\bar{{\bf X}} - \bar{{\bf Y}},\widehat{\psi}_{k} - \widehat{c}_{k}\psi_{k}\rangle\right| = o_{P}(1)  \label{eq7.1}
\end{eqnarray}
as $m, n \rightarrow \infty$ under this sequence of shifts. Here $\widehat{\psi}_{k}$ is the empirical version of $\psi_{k}$ and $\widehat{c}_{k} = sign(\langle\widehat{\psi}_{k},\psi_{k}\rangle)$. The limiting distribution of $mnN^{-1} \sum_{k=1}^{L} (\langle\bar{{\bf X}} - \bar{{\bf Y}},\widehat{\psi}_{k}\rangle)^{2}$ is the same as that of $mnN^{-1} \sum_{k=1}^{L} (\langle\bar{{\bf X}} - \bar{{\bf Y}},\widehat{c}_{k}\psi_{k}\rangle)^{2}$ in view of (\eqref{eq7.1}). Since the $\widehat{c}_{k}$'s take values $\pm1$ only, $mnN^{-1} \sum_{k=1}^{L} (\langle\bar{{\bf X}} - \bar{{\bf Y}},\widehat{c}_{k}\psi_{k}\rangle)^{2} = mnN^{-1}||{\bf v}||^{2}$, and the latter converges {\it weakly} to $||N_{L}(\boldsymbol\beta,\Lambda_{L})||^{2}$ as $m, n \rightarrow \infty$.
Thus, $mnN^{-1}T_{HKR1}$ converges {\it weakly} to $\sum_{k=1}^{L} \lambda_{k}\chi^{2}_{(1)}(\beta_{k}^{2}/\lambda_{k})$ under the given sequence of shrinking shifts as $m, n \rightarrow \infty$.  \\
\indent It also follows using similar arguments as in the proof of Theorem 5.3 in \cite{HK12} that under the assumptions of the present theorem, and for the given sequence of shrinking shifts, we have
\begin{eqnarray*}
&& \max_{1 \leq k \leq L} (mn/N)^{1/2} \widehat{\lambda}_{k}^{-1/2}\left|\langle\bar{{\bf X}} - \bar{{\bf Y}},\widehat{\psi}_{k} - \widehat{c}_{k}\psi_{k}\rangle\right| = o_{P}(1)
\end{eqnarray*}
as $m, n \rightarrow \infty$. Similar arguments as in the case of $T_{HKR1}$ now yield the asymptotic distribution of $mnN^{-1}T_{HKR2}$, and this completes the proof.
\end{proof}

\bibliographystyle{apalike.bst}
\bibliography{biblio-file.bib}

\end{document}